\documentclass[twocolumn]{IEEEtran}
\usepackage{amsfonts,amssymb,mathrsfs,amsmath,mathrsfs}

\usepackage{graphicx,color}

\usepackage[breaklinks=true]{hyperref}

\newtheorem{theorem}{Theorem}
\newtheorem{lemma}[theorem]{Lemma}
\newtheorem{proposition}[theorem]{Proposition}
\newtheorem{corollary}[theorem]{Corollary}

\newtheorem{example}{Example}

\newtheorem{remark}{Remark}

\newcommand{\beq}{\begin{equation}}
\newcommand{\eeq}{\end{equation}}
\newcommand{\bqa}{\begin{eqnarray}}
\newcommand{\eqa}{\end{eqnarray}}

\definecolor{green}{rgb}{0.00,0.50,0.00}

\newenvironment{proof}[1][Proof]{\noindent\textbf{#1.} }{\ \rule{0.5em}{0.5em}}

\begin{document}
\title{Classical and Quantum Stochastic Models of Resistive and Memristive Circuits}
\author{%
\IEEEauthorblockN{John E. Gough\IEEEauthorrefmark{1},
Guofeng Zhang\IEEEauthorrefmark{2},~\IEEEmembership{Member,~IEEE}} \\
\IEEEauthorblockA{\IEEEauthorrefmark{1}Department of Physics, Aberystwyth University, 
Ceredigion SY23 3BZ, Wales UK}\\
\IEEEauthorblockA{\IEEEauthorrefmark{2}Department of Applied Mathematics, 
The Hong Kong Polytechnic University, 
Hong Kong}\thanks{
Corresponding author: J.E. Gough (email: jug@aber.ac.uk).}}
\date{\today}
\maketitle


\IEEEtitleabstractindextext{
\begin{abstract}
The purpose of this paper is to examine stochastic Markovian models for circuits in phase space for which the drift term
is equivalent to the standard circuit equations. In particular we include dissipative components corresponding to
both a resistor and a memristor in series. We obtain a dilation of the problem for which is canonical in the sense
that the underlying Poisson Brackets structure is preserved under the stochastic flow. We do this first of all for
standard Wiener noise, but also treat the problem using a new concept of symplectic noise where the Poisson structure
is extended to the noise as well as the circuit variables, and in particular where we have canonically conjugate noises.
Finally we construct a dilation which describes the quantum mechanical analogue. 
\end{abstract}

}

\IEEEdisplaynontitleabstractindextext
\IEEEpeerreviewmaketitle

\section{Introduction}\label{sec:introduction}

 \IEEEPARstart{D}{issipation} has long been realized
as a feature, though also a resource, in the design of engineered systems, 
\cite{Wil72}-\cite{DS13}. The question of how exactly to model dissipation, given that the fundamental
physical dynamical equations of motion are Hamiltonian, remains of
fundamental importance; the issue applies to both classical and quantum
systems, \cite{Z73}-\cite{SDD2011}. Hamiltonian
systems with a finite number of degrees of freedom have dynamical evolutions
that preserve the canonical structure - that is, the Poisson brackets in
classical theory, and the commutation relations in quantum theory. In this
setting, dissipation is \textit{per se} impossible. There exist a number of
ingenious approaches to tackling this problem such as approximations using
lossless systems \cite{SDD2011} and redefining the Poisson bracket to
include dissipative effects \cite{vanderSchaft}. 

The approach adopted in this paper in order to introduce dissipation 
is to embed the system in an environment (leading to a joint coupled Hamiltonian model) and average out
the environment. In other words, we consider stochastic models which are Hamiltonian in structure
which dilate the dissipative dynamics. It is well-known that under various assumptions about the bulk limit of the
environment, negligible autocorrelation of the environment processes
(memoryless property), ignoring rapid oscillations (the rotating wave
approximation in the quantum case), etc., one may obtain limiting dynamical
models with an irreversible semi-group evolution. Moreover, this semi-group
may often be dilated to a stochastic Markov system: in the classical case,
this may be described by stochastic differential equations of motion, for
instance for a diffusion process, where the generator is the second-order
differential operator on phase space coinciding with the generator of the
semi-group; in the quantum setting, this may be a unitary quantum stochastic
process, leading to an evolution described by quantum stochastic calculus 
\cite{HP84}, where the generator is Lindbladian \cite{Lindblad,GKS}.

The purpose of this paper is to carry out this programme in the setting of
electronic circuit models, where we allow for dissipation beyond the usual
ohmic damping. It is well known that constant inductance-capacitance ($LC$)
circuits are Hamiltonian, and are readily quantized, \cite[Section 3.4.3]%
{GarZol00}. Markovian models that include dissipation may then be formulated
on purely phenomenological grounds, however, it is clear that it is
desirable to have a theory that is capable of dealing with non-linear
dissipation, in particular, we wish to include charge and current dependent
resistances, which brings us to the formalism introduced by Chua \cite%
{Chua71} for a unified axiomatic description of passive circuit
components. While originally a theoretical construct \cite{Chua71}-\cite{Cohen}, there is evidence to
suggest that ideal memristors also occur in quantum models \cite{YLI14}.

The layout and main contributions of this paper are summarized as follows. In
section \ref{sec:General Circuits} we recall the framework of Chua for
describing general classical circuit models and present the canonical formulation
in phase space. In section \ref{sec:Stochastic_Models} we show how we may obtain a stochastic dilation of
these models for an arbitrary resistor and memristor in series, and remarkably, we are able to give explicit constructions
based on the principle that the dilation should in some sense be Hamiltonian: more exactly,
this emerges from the requirement that the stochastic evolution be canonical.
In the theory of quantum stochastic evolutions \cite{HP84}, this turns out to arise automatically from the requirement of a
unitarity of the evolution process. The classical theory is more flexible,
but less transparent as one has to put in the requirement of preservation of the Poisson brackets as an additional constraint
on the stochastic evolution explicitly. Here the algebraic features encoded in the
quantum case are replaced by geometric features imposed on the diffusion
process.
Our first construction (Theorem \ref{Thm:main}) deals with the classical problem and utilizes a result by one of the authors in \cite{Gough99}.
This however involves standard Wiener noise which has no symplectic structure of its own leading to a rather involved form of the dilation.
Motivated by the situtation that occurs in quantum models, we consider pairs of Wiener noise that satisfy their own non-trivial Poisson
bracket relations: these symplectic noises \cite{Gough_JSP_2015} were only recently introduce and they parallel the corresponding situation in quantum theory where the
noise is actually a quantum electromagnetic field. We give the corresponding
result, Theorem \ref{Thm:main2}, for symplectic noise and this constitutes a genuine extension of Theorem \ref{Thm:main} as we obtain
the later in the case where only one of each of the canonical pairs of
symplectic noise is present. The construction in Theorem \ref{Thm:main2} to obtain the general canonical dilation of a resisor and memristor in series
is markedly more simple and natural than that of Theorem \ref{Thm:main} due to the role of the canonically conjugate noises.
For completeness, in section \ref{sec:Quantum} we show how these
models are readily quantized and establish the analogous result, Theorem \ref%
{Thm:main3}. In Section \ref{sec:appro} we discuss the origin of the stochastic models proposed in Sections 
\ref{sec:Stochastic_Models} and \ref{sec:Quantum} . Finally we conclude the paper in Section \ref{sec:conclusion}.
Finally, there has been much interest into how Hamiltonian systems can approximate dissipatve systems, see \cite{Spohn} for classical 
and \cite{AFL},\cite{AGL} for quantum, and there have been a control theoretic models of lossless approximations
to dissipative systems \cite{SDD2007}, \cite{SDD2011}. In section \ref{sec:appro}, we show how these models may arise from 
approximations by physically realistic microscopic systems.

\section{General Circuit Theory}

\label{sec:General Circuits} The concept of a memristor was introduced by
Leon Chua in 1971 as a missing fourth element in the theory of idealized
passive components in electronics \cite{Chua71,ChuaKang76}. There are four
variables to consider in a circuit: the charge $q$, the current $I$, the
flux $\varphi $, and the voltage $V$. They are actually paired naturally as $%
(q,I)$ and $(\varphi , V)$, and we always have the fundamental relations 
\begin{equation}
I=\dot{q},\quad V=\dot{\varphi},  \label{eq:dots}
\end{equation}
however, it is mathematically useful to treat these as four independent
variables.

An ideal \textit{resistor} $\mathscr{R}$ is a component that leads to a
fixed $I$-$V$ characteristic, that is, the voltage $V$ across the resistor
and the current $I$ through the resistor are constrained by an equation of
the form 
\begin{equation*}
f_{\mathscr{R}}\left( I,V\right) =0,
\end{equation*}
and assuming that the relation is differentiable and bijective (that is,
each voltage determines a unique current and vice versa) we may always write 
\begin{equation*}
dV=R\left( I\right) \,dI.
\end{equation*}
Note that we can always choose $R$ to be a function of $I$ only by
assumption. The coefficient $R$ is the resistance, and in the case where it
is a constant $R$ we obtain Ohm's law $V=RI$.

Likewise, an ideal \textit{capacitor} $\mathscr{C}$ is a component that
determines a $q$-$V$ characteristic $f_{\mathscr{C}}\left( q,V\right) =0$
and under similar assumptions we may write 
\begin{equation*}
dV=\frac{1}{C\left( q\right) }\,dq;
\end{equation*}
while an \textit{inductor} $\mathscr{L}$ is a component that determines a $%
\varphi$-$I$ characteristic $f_{\mathscr{L}}\left( \varphi ,I\right) $ and
under similar assumptions we have 
\begin{equation*}
d\varphi =L\left( I\right) \,dI.
\end{equation*}
The physical quantities $C$ and $L$ arising are the capacitance and
inductance respectively, and it is convenient to take them as functions of $%
q $ and $I$ respectively.

Chua's insight was that the mathematical theory missed out an idealized
component $\mathscr{M}$ that determined a $q$-$\varphi $ characteristic.
This fourth element he called a \textit{memristor} and it fixed a
relationship $f_{\mathscr{M}} \left( q,\varphi \right) =0$, and again
assuming differentiability and bijective we are led to the differential
relation 
\begin{equation*}
d\varphi =M\left( q\right) \,dq.
\end{equation*}
In fact, it immediately follows from (\ref{eq:dots}) that for a memristor we
have the relation 
\begin{equation*}
V=M\left( q\right) \,I,
\end{equation*}
and comparison with Ohm's law suggests that the physical quantity $M\left(
q\right) $ is a charge dependent resistance. Indeed as the memristor relates
the integral of the voltage $\varphi \left( t\right) =\int_{-\infty
}^{t}V\left( t^{\prime }\right) dt^{\prime }$ to the integral of the current 
$q\left( t\right) =\int_{-\infty }^{t}I\left( t^{\prime }\right) dt^{\prime
} $, the component acts as a resistive element with memory of past voltages
and currents - whence the name memristor.

The capacitor and inductor are both capable of storing energy and we have
the associated energies 
\begin{eqnarray*}
E_{\mathscr{C}} &=&\int qdV=\int \frac{qdq}{C\left( q\right) }, \\
E_{\mathscr{L}} &=&\int Id\varphi =\int L\left( I\right) I dI.
\end{eqnarray*}

We may summarize the main properties of the four components as follows:

\bigskip

\textbf{The Inductor} $\mathscr{L}$

\bigskip

\begin{tabular}{l||l}
Type & Energy Storing \\ 
Characteristic & $d\varphi =L(I)dI$ \\ 
Voltage & $V_{\mathscr{L}}=\dot{\varphi}_{\mathscr{L}}=L(I_{\mathscr{L}})%
\dot{I}_{\mathscr{L}}$%
\end{tabular}

\bigskip

\textbf{The Capacitor} $\mathscr{C}$

\bigskip

\begin{tabular}{l||l}
Type & Energy Storing \\ 
Characteristic & $dV=\frac{1}{C\left( q\right) }dq$ \\ 
Voltage & $V_{\mathscr{C}}=\left. \int \frac{dq}{C\left( q\right) }\right|
_{q_{\mathscr{L}}}$%
\end{tabular}

\bigskip

\textbf{The Resistor} $\mathscr{R}$

\bigskip

\begin{tabular}{l||l}
Type & Dissipative \\ 
Characteristic & $dV=R(I)dI$ \\ 
Voltage & $V_{\mathscr{R}}=\left. \int R(I)dI\right| _{I_{\mathscr{R}}}$%
\end{tabular}

\bigskip

\textbf{The Memristor} $\mathscr{M}$

\bigskip

\begin{tabular}{l||l}
Type & Dissipative \\ 
Characteristic & $d\varphi =M(q)dq$ \\ 
Voltage & $V_{\mathscr{M}}=\dot{\varphi}_{\mathscr{M}}=M(q_{\mathscr{M}}) I_{%
\mathscr{M}}$%
\end{tabular}

\bigskip

These four components will then form the building blocks for general passive
electric circuits. While they are clearly idealizations, it is nevertheless
the case that they are the fundamental elements that can be combined to
produce physical systems.

\begin{figure}[tbph]
\centering
\includegraphics[width=0.40\textwidth]{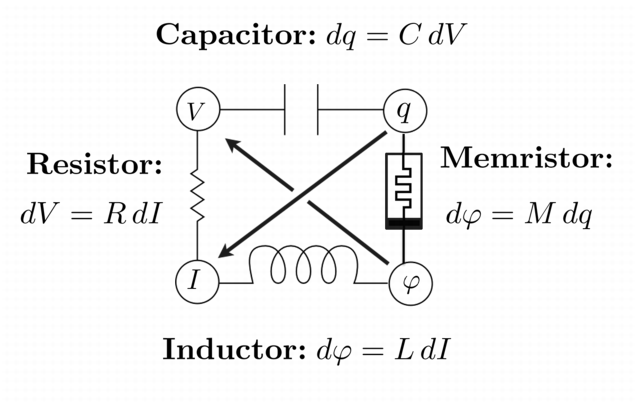}
\caption{Chua's fourfold way! The inclusion of the memristor conceptually
completes the set of idealized passive circuit elements.}
\label{fig:4_elements}
\end{figure}

The introduction of the memristor by Chua has the important theoretical
implication that we can model a wider class of dissipative systems than just
described by the resistor alone. We will explore this shortly, but the main
feature is that we may have dissipative models where the damping is
dependent on the charge in addition to the current.

\subsection{Lagrangian \& Hamiltonian models for energy storing components}

The situation where we have an inductor and capacitor in series is
well-known to be a conservative system with a Lagrangian formulation, see
also \cite{Jelstema,DiVentra,Cohen}. The equation of motion follows from $V_{%
\mathscr{L}}+V_{\mathscr{C}}=e\left( t\right)$ and we are 
\begin{equation*}
L(\dot{q}) \, \ddot{q} + \Phi_{\mathscr{C}}^\prime (q) = e(t),
\end{equation*}
with 
\begin{equation*}
\Phi_{\mathscr{C}}^\prime (q) = \int \frac{dq}{C\left( q\right) } .
\end{equation*}
We assume the existence of a twice-differentiable function $K$ such that $%
K^\prime (I) >0$ and $K^{\prime \prime} (I) = L(I)$. The equations of the $%
\mathscr{LC}$-circuit then follow from the Euler-Lagrange equations $\frac{d%
}{dt} ( \frac{\partial \mathcal{L}}{\partial \dot{q}} ) -\frac{\partial 
\mathcal{L}}{\partial q}=0$ with Lagrangian 
\begin{equation*}
\mathcal{L}(q,\dot{q}, t) = K(\dot{q}) - \Phi_{\mathscr{C}} (q) +e(t) q .
\end{equation*}
The canonical momentum is then given by 
\begin{equation*}
p= \frac{\partial \mathcal{L} }{\partial \dot{q}} = K^\prime ( \dot{q} ) ,
\end{equation*}
and by assumption on $K$ this is a bijection with inverse $\dot{q} = \mathbb{%
I}(p)$, that is, the Lagrangian is hyper-regular. Note that 
\begin{equation*}
\dot{p} = L( \dot{q} ) \, \ddot{q} = L(I) \, \dot{I} \equiv V_{\mathscr{L}}= 
\dot{\varphi}_{\mathscr{L}},
\end{equation*}
so in a sense $p$ may be identified with the inductor flux $\varphi_{%
\mathscr{L}}$. The corresponding Hamiltonian is then 
\begin{eqnarray}
H(q,p,t) = \mathbb{K} (p) + \Phi_{\mathscr{C}} (q) -e(t) q,  \label{eq:Ham}
\end{eqnarray}
and we have the Legendre transform 
\begin{equation*}
\mathbb{K}(p) = \sup_{I} \{ pI -K(I) \} \equiv p \mathbb{I} (p) - K (\mathbb{%
I} (p)).
\end{equation*}

In the special case where the inductance is $L(I)=L_0$ (constant) we take $%
K(I) = \frac{1}{2} L_0 I^2$, and here 
\begin{equation*}
\mathbb{I} (p) = p/L_0, \quad \mathbb{K} (p ) =\frac{1}{2L_0}p^2.
\end{equation*}
Likewise, in the special case where the capacitance is $C(q) =C_0$, we may
take $\Phi_{\mathscr{C}} (q)= \frac{1}{2C_0}q^2$. Here, the Hamiltonian
takes the explicit form 
\begin{equation}
H_0 (q,p,t) = \frac{p^2}{2L_0} + \frac{q^2}{2C_0} - q\, e(t),  \label{eq:H0}
\end{equation}
with $p=L_0 \dot q$, and that this is a driven harmonic oscillator with
resonant frequency $\omega_0 = (L_0 C_0)^{-1/2 }$. Note that, with $e(t)=0$, 
$H_0$ is equal to the physically stored energy $E_{\mathscr{L}} +E_{%
\mathscr{C}}$ with the appropriate substitution $\dot{q} = p/L_0$.

\subsection{Dissipative Circuits}

We begin by examining a circuit having all four components in series, see
Figure \ref{fig:LCRM_circuit}.

\begin{figure}[htbp]
\centering
\includegraphics[width=0.400\textwidth]{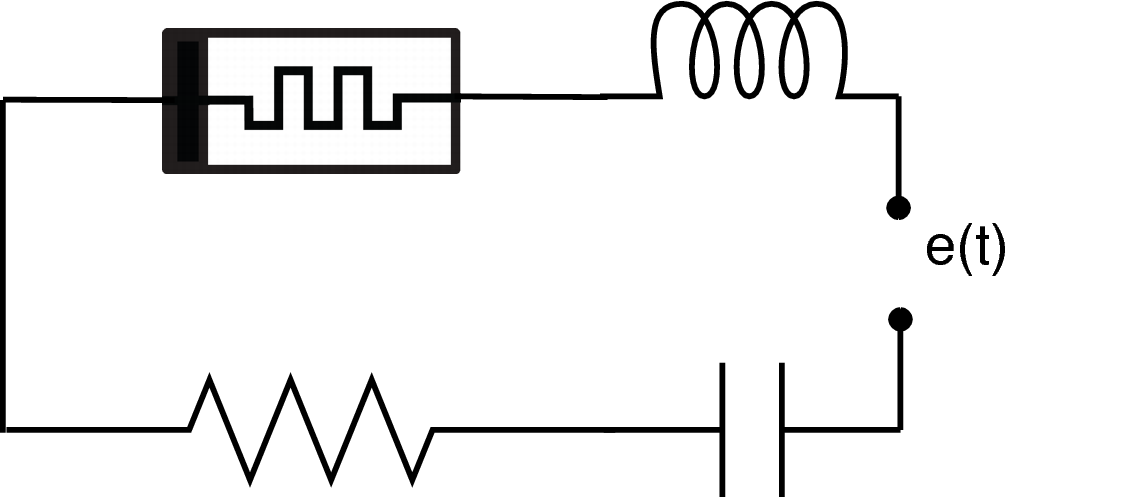}
\caption{All four idealized components in series.}
\label{fig:LCRM_circuit}
\end{figure}

The applied emf $e(t)$ driving the circuit is related by Kirchhoff's voltage
law to the total voltage drop over the circuit: 
\begin{equation*}
V_{\mathscr{L}}+V_{\mathscr{C}}+V_{\mathscr{R}}+V_{\mathscr{M}}=e\left(
t\right).
\end{equation*}

Recalling that $V_{\mathscr{L}}=\dot{p}$, the dynamical equation may instead
be expressed as 
\begin{equation*}
\dot{p}=-\int \frac{dq}{C\left( q\right) }-\int R(I)dI- M\left( q\right) 
\dot{q} +e\left( t\right) .
\end{equation*}
Let us introduce the functions 
\begin{eqnarray}
\mathcal{D} _{\mathscr{R}}^{\prime }(I) &=& \int R(I) dI ,  \notag \\
\Psi _{\mathscr{R}}^{\prime }(p) &=& \mathcal{D} _{\mathscr{R}}^{\prime } ( 
\mathbb{I} (p)),  \label{eq:dissip_R}
\end{eqnarray}
then we obtain the system of equations in the $\left( q,p\right) $ phase
space of the circuit 
\begin{eqnarray}
\dot{q} &=& \mathbb{I} (p),  \notag \\
\dot{p} &=&-\Phi _{\mathscr{C}}^{\prime }\left( q\right) -\Psi _{\mathscr{R}%
}^{\prime }(p)- M\left( q\right)\mathbb{I}(p) +e\left( t\right) .
\label{eq:base}
\end{eqnarray}
The negative divergence of the velocity field of phase points is then 
\begin{eqnarray}
\gamma \left( q,p\right) &=&-\left( \frac{\partial }{\partial q}\dot{q}+%
\frac{\partial }{\partial p}\dot{p}\right)  \nonumber \\
&=&\Psi _{\mathscr{R}}^{\prime \prime }(p)+M\left( q\right)\frac{\partial 
\mathbb{I}}{\partial p}  \nonumber\\
&=&\frac{1}{\mathbb{L} (p) } \left( \mathbb{R} (p) + M\left( q\right)
\right) , \label{eq:_secII_gamma}
\end{eqnarray}
where we define the $p$-dependent resistance and inductance as 
\begin{equation*}
\mathbb{R} (p)= R (\mathbb{I} (p)) , \quad \mathbb{L} (p)= L (\mathbb{I}
(p)).
\end{equation*}

The equations of motion are 
\begin{eqnarray}
\dot{q} &=& \frac{\partial H}{\partial p } , \\
\dot{p} &=& -\frac{\partial H}{\partial q } - \mathbb{V}_{\mathscr{R}} (p) - 
\mathbb{V}_{\mathscr{M}} (q,p) ,
\end{eqnarray}
with $H$ the Hamiltonian (\ref{eq:Ham}) for the energy storing components $%
\mathscr{L}$ and $\mathscr{C}$ and the voltages associated with the energy
dissipating components $\mathscr{R}$ and $\mathscr{M}$ are 
\begin{eqnarray}
\mathbb{V}_{\mathscr{R}} (p) &=&\left. \int R(I)dI\right| _{I= \mathbb{I}
(p)}, \; \text{i.e.} \, \Psi _{\mathscr{R}}^{\prime }(p), \\
\mathbb{V}_{\mathscr{M}} (q,p) &=& M\left( q\right)\mathbb{I}(p) .
\end{eqnarray}

In the situation where the inductance is a constant $L_0$ the system of
equations (\ref{eq:base}) reduce to $\dot{q} = v^q, \, \dot{p} = v^p$ with 
\begin{eqnarray}
v^{q} &=& \frac{p}{L_0},  \notag \\
v^{p} &=&-\Phi _{\mathscr{C}}^{\prime }\left( q\right) -\Psi _{\mathscr{R}%
}^{\prime }(p)-\frac{1}{L_0}M\left( q\right)p+e\left( t\right) .
\label{eq:drift_velocity}
\end{eqnarray}

\begin{remark}
we shall generally be interested in the behaviour of the Poisson brackets 
\begin{eqnarray}
\{ f,g \}= \frac{\partial f }{\partial q} \frac{\partial g}{\partial p}- 
\frac{\partial g }{\partial q} \frac{\partial f}{\partial p}
\end{eqnarray}
under general dynamical flows on phase space.

In general, let $w^q$ and $w^p$ be twice-differentiable coefficients of a
vector field in phase space along the coordinate axes, then the tangent
vector field is the operator 
\begin{equation*}
w = w^q \frac{\partial }{\partial q}+ w^p \frac{\partial }{\partial p}.
\end{equation*}
We say that the field is Hamiltonian if it can be written in terms of the
Poisson brackets as $w( \cdot ) \equiv \{ \cdot , H\}$, for some smooth
function $H$. In coordinate form this reads as $w^q =\frac{\partial H}{%
\partial p}, \, w^p =- \frac{\partial H}{\partial q}$.

Given a general vector field $w$, the flow it generates is the family of
diffeomorphisms $\{ \phi_t \}$, for $t$ in a neighbourhood of 0, such that $%
t\mapsto (q_t,p_t)\equiv \phi_t (q_0,p_0)$ gives the integral curve of $w$
passing through $(q_0,p_0)$ at $t=0$. The flow is canonical if it preserves
the Poisson brackets. To obtain a necessary condition, we look at small
times $t$ where locally we have 
\begin{eqnarray*}
q_{t} &=&q_0+t v^{q}\left( q_0,p_0\right) +O\left( t^{2}\right) \\
p_{t} &=&p_0+t v^{p}\left( q_0,p_0\right) +O\left( t^{2}\right)
\end{eqnarray*}
and so 
\begin{eqnarray*}
\frac{\partial q_{t}}{\partial q} &=&1+tv_{\,,q}^{q}+O\left( t^{2}\right) ,%
\frac{\partial q_{t}}{\partial p}=tv_{\,,p}^{q}+O\left( t^{2}\right) \\
\frac{\partial p_{t}}{\partial q} &=&tv_{\,,q}^{p}+O\left( t^{2}\right) ,%
\frac{\partial p_{t}}{\partial p}=1+tv_{\,,p}^{p}+O\left( t^{2}\right)
\end{eqnarray*}
We therefore see that 
\begin{eqnarray*}
\left\{ q_{t},p_{t}\right\} &=&1+t\left( v_{\,,q}^{q}+v_{\,,p}^{p}\right)
+O\left( t^{2}\right) \\
&=&1 - \gamma (q_0,p_0) t +O\left( t^{2}\right).
\end{eqnarray*}
A necessary condition for the flow to be canonical is therefore that $\gamma
(q,p) \equiv 0$, and, as is well-known, the possible solutions take the form 
$v^{q}=\frac{\partial H}{\partial p},v^{p}=-\frac{\partial H}{\partial q}$
for some function $H$, in other words $w$ must be a Hamiltonian vector
field. More generally, the infinitesimal expression describing the
preservation of the Poisson brackets under the flow integral to $w$ is 
\begin{equation}
w \left( \{ f,g \} \right) = \{ w(f) , g \} + \{ f, w(g) \}  \label{eq:can}
\end{equation}
for every pair of smooth functions $f,g$ on phase space. Again it can be
shown that this holds if and only if $w$ is divergence free, that is, $\frac{%
\partial }{\partial q}w^q+ \frac{\partial }{\partial p}w^p =0$, and this of
course is equivalent to $w$ being a Hamiltonian vector field.

In general, the dissipation $\gamma (q,p)$ gives the exponential rate at
which phase area in the $qp$ phase space is decreasing. As physically $R\geq
0$ and $M\geq 0$ for passive systems, we must have $\gamma \geq 0$.
Geometrically, $\gamma$ characterizes exactly the non-Hamiltonian nature of
the dynamical flow on phase space, and we conclude that a dynamical flow on
phase space will preserve the Poisson brackets if and only if it corresponds
to the evolution governed by some Hamiltonian function $H$.
\end{remark}

\bigskip

\begin{remark}
For the systems in series as above, we may decompose 
\begin{equation*}
\gamma (q,p) = \gamma_{\mathscr{R}} (p) + \gamma_{\mathscr{M}} (q,p)
\end{equation*}
with $\gamma_{\mathscr{R}} (p) = \frac{\partial }{\partial p}\mathbb{V}_{%
\mathscr{R}}$ and $\gamma_{\mathscr{M}} (q,p) = \frac{\partial }{\partial p}%
\mathbb{V}_{\mathscr{M}}$. In the case where the inductances are constant we
find 
\begin{equation*}
\gamma_{\mathscr{R}} (p) =\frac{1}{L_0} \mathbb{R} (p), \quad \gamma_{%
\mathscr{M}} (q,p) =\frac{1}{L_0} M(q),
\end{equation*}
and in particular these are functions of $p$ only and, respectively, $q$
only. Therefore we have a natural decomposition (up to an additive constant)
of the resistor and memristor elements as providing the $p$ only and $q$
only contributions for the dissipation $\gamma$, respectively, when the
inductances are constant.
\end{remark}

\bigskip

\begin{remark}
It is fairly easy to construct dissipative components that do not have the
special decomposition $\gamma (q,p) = \gamma_{\mathscr{R}} (p) + \gamma_{%
\mathscr{M}} (q)$ for constant inductance. A simple example is a resistor
and memristor in parallel, where now 
\begin{equation*}
\mathbb{V} (q,p) = \left( \mathbb{V}_{\mathscr{R}}(p)^{-1} +\mathbb{V}_{%
\mathscr{M}}(q)^{-1} \right)^{-1},
\end{equation*}
see Figure \ref{fig:RM_parallel}. 
\begin{figure}[h]
\centering
\includegraphics[width=0.250\textwidth]{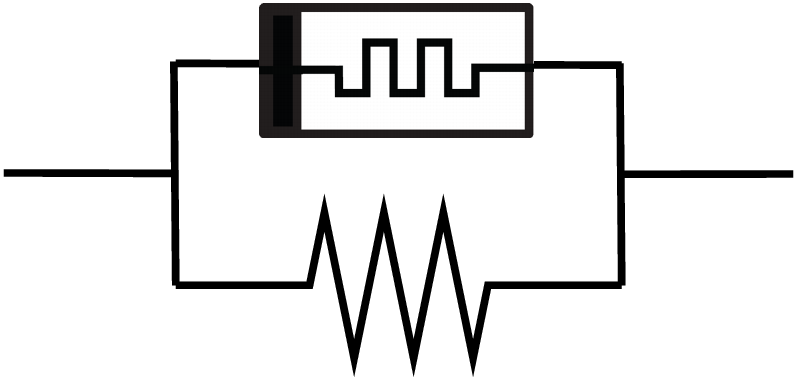}
\caption{A resistor and memristor in parallel.}
\label{fig:RM_parallel.png}
\end{figure}

Without going into explicit details, it is intuitively obvious that an
arbitrary damping function $\gamma (q,p) \geq 0$ may be approximated by a
network of resistive and memristive components in series and parallel.
\end{remark}

\bigskip

\begin{remark}
The dynamical equations may also be written in Lagrangian terms as 
\begin{equation*}
\frac{d}{dt} \left( \frac{\partial \mathcal{L}}{\partial \dot{q}} \right) -%
\frac{\partial \mathcal{L}}{\partial q}= -\frac{\partial \mathcal{D}}{%
\partial \dot{q}},
\end{equation*}
where we must now introduce the so-called dissipation potential $\mathcal{D}
(q,\dot{q}) = \mathcal{D}_{\mathscr{R}} ( \dot{q} ) + \frac{1}{2} M(q) \dot{q%
}^2$ which is a sum of the resistance potential $\mathcal{D}_{\mathscr{R}}$
introduced in (\ref{eq:dissip_R}) and a memristive Rayleigh type dissipation
function. However, we shall work in the Hamiltonian formalism.
\end{remark}

\subsection{Equivalent Circuit Theorem}

Based on the above discussions, we have the following separation of
arbitrary passive circuits into idealized energy storing and energy
dissipating components.

\begin{theorem}
An arbitrary passive circuit may be naturally decomposed in an energy
storing components, the \textbf{Hamiltonor}, described by a Hamilton's
function $H$ and an energy dissipating component, the \textbf{Dissipator},
described by a voltage function $\mathbb{V}_{\mathscr{D}}$, see Figure \ref%
{fig:HD_circuit}, such that the circuit equations are 
\begin{eqnarray}
\dot{q} &=&\frac{\partial H}{\partial p}, \\
\dot{p} &=&-\frac{\partial H}{\partial q}-\mathbb{V}_{\mathscr{D}}(q,p).
\label{eq:gen}
\end{eqnarray}
The dissipation is then the negative divergence of the phase velocity, and
this is given by 
\begin{equation*}
\gamma (q,p)=\frac{\partial }{\partial p}\mathbb{V}_{\mathscr{D}}(q,p)\geq 0.
\end{equation*}
\end{theorem}

\begin{figure}[h]
\centering
\includegraphics[width=0.480\textwidth]{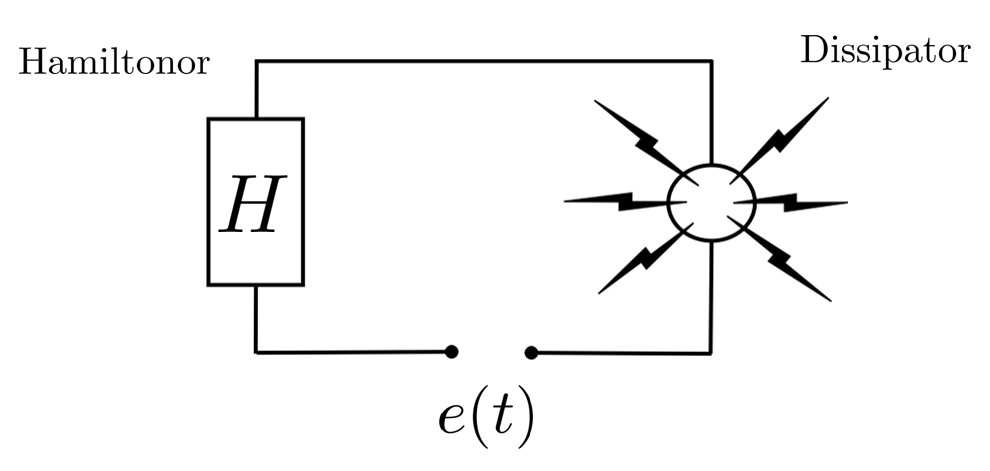}
\caption{A Hamiltonian component (Hamiltonor) and dissipative component
(dissipator) in series.}
\label{fig:HD_circuit.png}
\end{figure}

From the mathematical perspective, the beauty of Chua's introduction of the
memristor as a dissipative element is that it enlarges the class of possible
dissipators in a maximal way.

\section{Classical Stochastic Models}

\label{sec:Stochastic_Models}

\subsection{Diffusion Noise Models}

We now consider a stochastic dynamics on phase space given by the
Stratonovich stochastic differential equation (SDE)
\begin{equation}
dq = w^q \, dt + \sum_\alpha \sigma_\alpha^q \circ dB^\alpha_t, \quad dp =
w^p \, dt +\sum_\alpha \sigma_\alpha^p \circ dB^\alpha_t,  \label{eq:strat}
\end{equation}
where the coefficients $w^q,w^p,\sigma_\alpha^q,\sigma_\alpha^p$ are assumed
twice differentiable functions of the coordinates $(q,p)$ and satisfy
Lipschitz and growth conditions to ensure existence and uniqueness of
solution \cite{Oksendal}. Here we take independent standard Wiener processes 
$(B^\alpha_t)_{t\geq 0}$. The equation may be cast in It\={o} form as 
\begin{eqnarray*}
dq &=& \left( w^q + \frac{1}{2} \sum_\alpha \sigma_\alpha^q \frac{\partial
\sigma_\alpha^q}{\partial q} + \frac{1}{2} \sum_\alpha\sigma_\alpha^p \frac{%
\partial \sigma_\alpha^q}{\partial p} \right) \, dt \\
&&+ \sum_\alpha \sigma_\alpha^q \, dB^\alpha_t, \\
dp &=& \left( w^p + \frac{1}{2}\sum_\alpha \sigma_\alpha^q \frac{\partial
\sigma_\alpha^p}{\partial q} + \frac{1}{2}\sum_\alpha \sigma_\alpha^p \frac{%
\partial \sigma_\alpha^p}{\partial p} \right) \, dt \\
&&+ \sum_\alpha \sigma_\alpha^p \, dB^\alpha_t.
\end{eqnarray*}

We now require that a stochastic differential flow preserves the Poisson
structure in the sense that we now obtain a random family of canonical
diffeomorphisms on phase space. In the It\={o} calculus, this implies that 
\begin{equation*}
d \{ f_t,g_t \}_t = \{ d f_t,g_t \}_t +\{ f_t, dg_t \}_t + \{ df_t ,dg_t
\}_t .
\end{equation*}
We now quote the main result from \cite{Gough99}, section III.

\bigskip

\begin{theorem}
\label{thm:JMP1999} A stochastic dynamics determined by the system of
stochastic differential equations (\ref{eq:strat}) on phase space will be
canonical if and only if the vector fields $w$ and $\sigma_\alpha$ are
Hamiltonian.
\end{theorem}

We now suppose that this is the case and take $w^q =\frac{\partial H}{%
\partial p}, \, w^p =- \frac{\partial H}{\partial q}$ and $\sigma_\alpha^q =%
\frac{\partial F_\alpha}{\partial p}, \, \sigma_\alpha^p =- \frac{\partial
F_\alpha}{\partial q}$. In the It\={o} calculus we find the following
Langevin equation for functions $f_t =f(q_t,p_t)$ 
\begin{eqnarray*}
d f_t & =& \left( \{ f, H \}_t + \frac{1}{2}\sum_\alpha \{ \{f, F_{\alpha }
\}_t , F_{\alpha } \}_t \right) \, dt \\
&& + \sum_\alpha \{ f, F_{\alpha } \}_t \, dB^\alpha_t .
\end{eqnarray*}
The dissipative component of the evolution is described by the double
Poisson bracket with respect to the $F_\alpha$ in the drift ($dt$-term).
This is the generator of the diffusion, which in fact is the second-order
differential operator 
\begin{eqnarray}
\mathscr{L} = \{ \cdot , H \} + \frac{1}{2}\sum_\alpha \{ \{ \cdot ,
F_{\alpha } \} , F_{\alpha } \}.
\end{eqnarray}
We have for instance, 
\begin{equation*}
\mathscr{L} (q) = v^q, \quad \mathscr{L}(p) = v^p,
\end{equation*}
where 
\begin{eqnarray}
v^q = \frac{\partial H}{\partial p} + \frac{1}{2} \sum_\alpha\frac{\partial
F_\alpha}{\partial p} \frac{\partial^2 F_\alpha}{\partial q \partial p} - 
\frac{1}{2} \sum_\alpha \frac{\partial F_\alpha}{\partial q} \frac{%
\partial^2 F_\alpha}{\partial p^2} ,  \notag \\
v^p = -\frac{\partial H}{\partial q} - \frac{1}{2} \sum_\alpha \frac{%
\partial F_\alpha}{\partial p} \frac{\partial^2 F_\alpha}{\partial q^2} + 
\frac{1}{2} \sum_\alpha\frac{\partial F_\alpha}{\partial q} \frac{\partial^2
F_\alpha}{\partial p \partial q} .  \label{eq:itov}
\end{eqnarray}

 The It\={o}
equations for the coordinates become 
\begin{eqnarray*}
dq &=& v^q \, dt +\sum_\alpha \frac{\partial F_\alpha}{\partial p} \,
dB^\alpha_t, \\
dp &=& v^p \, dt - \sum_\alpha \frac{\partial F_\alpha}{\partial q} \,
dB^\alpha_t.
\end{eqnarray*}

Note that (\ref{eq:itov}) may be written more compactly as 
\begin{equation*}
v=J \nabla H + \frac{1}{2} \sum_\alpha J F^{\prime \prime}_\alpha J \nabla
F_\alpha ,
\end{equation*}
where the symplectic matrix $J$ and the Hessian $F^{\prime \prime}$ are
defined respectively as 
\begin{equation*}
J= \left[ 
\begin{array}{cc}
0 & 1 \\ 
-1 & 0%
\end{array}
\right] , \quad F^{\prime \prime} = \left[ 
\begin{array}{cc}
\frac{\partial^2 F }{\partial q^2} & \frac{\partial^2 F }{\partial q
\partial p} \\ 
\frac{\partial^2 F }{\partial q \partial p } & \frac{\partial^2 F }{\partial
p^2}%
\end{array}
\right] .
\end{equation*}

\begin{proposition}
The dissipation associated with the system of equations (\ref{eq:itov}) is
the sum of the Hessian determinants of the $F_\alpha$ 
\begin{eqnarray}
\gamma (q,p) = \sum_\alpha \left( \frac{\partial^2 F_\alpha }{\partial q^2} 
\frac{\partial^2 F_\alpha }{\partial p^2} - \left( \frac{\partial^2 F_\alpha 
}{\partial q \partial p} \right)^2 \right) .
\end{eqnarray}
\end{proposition}

\begin{proof}
This follows by substituting the expressions (\ref{eq:itov}) above for the It%
\={o} drift coefficients into $\gamma(q,p) = - \left( \frac{\partial v^q}{%
\partial q}+ \frac{\partial v^p}{\partial p}\right)$.
\end{proof}

The stochastic flow will, however, be canonical as the dissipation is
balanced geometrically by the fluctuations in the noise term. That is,
\begin{equation}
dq_t dp_t = -\sum_\alpha \frac{\partial F_\alpha}{\partial p} \frac{\partial F_\alpha}{\partial q}dt.
\end{equation}


The model is said to be \textit{passive} if we have dissipation $\gamma
(q,p) \geq 0$ at all phase points $(q,p)$, otherwise it is \textit{active}.
The van der Pol oscillator has $\gamma (q,p) = c (q^2 -a^2)$ for constants $%
a,c>0$ and is an example of an active model: this can be cast in the above
form with a single Wiener process, see \cite[Section V.B]{Gough99} for detail.

We would now like to realize models of the form (\ref{eq:gen}) as the
stochastic canonical models. That is, to obtain the velocity fields $v^q = 
\frac{\partial H}{\partial p} $ and $v^p = -\frac{\partial H}{\partial q} - 
\mathbb{V}_{\mathscr{D}} $, where $\mathbb{V}_{\mathscr{D}} (q,p) $ is the
voltage from the dissipative circuit elements, as the It\={o} drift of a
given stochastic canonical model. This is possible if we can solve the
following system of equations for $\{ F_\alpha \}$: 
\begin{eqnarray}
\sum_\alpha\frac{\partial F_\alpha}{\partial p} \frac{\partial^2 F_\alpha}{%
\partial q \partial p} - \sum_\alpha \frac{\partial F_\alpha}{\partial q} 
\frac{\partial^2 F_\alpha}{\partial p^2}=0 ,  \label{eq:0}\\
\sum_\alpha \frac{\partial F_\alpha}{\partial p} \frac{\partial^2 F_\alpha}{%
\partial q^2} - \sum_\alpha\frac{\partial F_\alpha}{\partial q} \frac{%
\partial^2 F_\alpha}{\partial p \partial q} =2\, \mathbb{V}_{\mathscr{D}} .
\label{eq:v}
\end{eqnarray}

In the following, we shall restrict to the problem of a single Wiener
process $B_t$ and function $F$.

\begin{lemma}
The equation (\ref{eq:0}) requires that there exists a function $\xi (q)$ such that $\xi (q) \frac{%
\partial F}{\partial p} = \frac{\partial F}{\partial q }$.
\end{lemma}

\begin{proof}
We may rewrite the equation (\ref{eq:0})  as $(\frac{\partial F}{\partial p} )^{-1} \frac{%
\partial }{ \partial p} (\frac{\partial F}{\partial p} ) - (\frac{\partial F%
}{\partial q} )^{-1} \frac{\partial }{ \partial p} (\frac{\partial F}{%
\partial q} )=0 $, and so 
\begin{equation*}
\frac{\partial }{ \partial p} \log \left( \frac{\partial F}{\partial q} / 
\frac{\partial F}{\partial p} \right) =0
\end{equation*}
and so $\frac{\partial F}{\partial q} / \frac{\partial F}{\partial p} = \xi
(q)$ for some function $\xi$ of $q$ only.
\end{proof}

We see that we may now substitute the identity $\xi (q) \frac{\partial F}{%
\partial p} = \frac{\partial F}{\partial q }$ into the second relation (\ref{eq:v}) to
get $\xi^\prime (q) ( \frac{\partial F}{\partial p})^2 = 2\, \mathbb{V}_{%
\mathscr{D}} $. However, we note that solving this partial differential equation for $F$ is in
general a difficult problem. We next show, remarkably, that there exist a
broad class of solutions to physically important cases.

\begin{theorem}
\label{Thm:main} A stochastic evolution with It\={o} drift velocity $%
(v^{q},v^{p})$ given by the phase velocity field (\ref{eq:drift_velocity})
can be achieved by a stochastic canonical model driven by a pair of
independent Wiener processes $(B_{t}^{1},B_{t}^{2})$, for the resistance and
memristance components respectively, with the following choices: 
\begin{eqnarray}
H &=&\frac{p^{2}}{2L_{0}}+\Phi _{\mathscr{C}}(q)+\frac{1}{2}W^{\prime }(p)q+%
\frac{1}{2}G^{\prime }(q)p-e(t)q,  \notag  \label{eq:hello} \\
F_{1} &=&\frac{q^{2}}{2c}+cW(p), \\
F_{2} &=&\frac{p^{2}}{2\ell }+\ell G(q),
\end{eqnarray}%
where $c$ and $\ell $ are constants with units of capacitance and
inductance, respectively, and we choose $W(p)$ and $G(q)$ such that 
\begin{equation}\label{eq:WG}
W^\prime(p)=\Psi^\prime _{\mathscr{R}}(p), ~~~ G^{\prime \prime }(q)=\frac{1}{L_{0}}M(q).
\end{equation}
\end{theorem}

\begin{proof}
Substituting in the expressions for the It\={o} drift in (\ref{eq:itov})
leads to 
\begin{eqnarray}
v^{q} &=&\frac{\partial H}{\partial p}+\frac{1}{2}\sum_{\alpha =1,2}\frac{%
\partial F_{\alpha }}{\partial p}\frac{\partial ^{2}F_{\alpha }}{\partial
q\partial p}-\frac{1}{2}\sum_{\alpha =1,2}\frac{\partial F_{\alpha }}{%
\partial q}\frac{\partial ^{2}F_{\alpha }}{\partial p^{2}},  \notag \\
&=&\left( \frac{p}{L_{0}}+\frac{1}{2}W^{\prime \prime }(p)q+\frac{1}{2}%
G^{\prime }(q)\right)  \notag \\
&&-\frac{1}{2}W^{\prime \prime }(p)q-\frac{1}{2}G^{\prime }(q)  \notag \\
&=&\frac{p}{L_{0}},  \notag \\
v^{p} &=&-\frac{\partial H}{\partial q}-\frac{1}{2}\sum_{\alpha =1,2}\frac{%
\partial F_{\alpha }}{\partial p}\frac{\partial ^{2}F_{\alpha }}{\partial
q^{2}}+\frac{1}{2}\sum_{\alpha =1,2}\frac{\partial F_{\alpha }}{\partial q}%
\frac{\partial ^{2}F_{\alpha }}{\partial p\partial q}  \notag \\
&=&-\left( \Phi _{\mathscr{C}}^{\prime }(q)+\frac{1}{2}W^{\prime }(p)+\frac{1%
}{2}G^{\prime \prime }(q)p-e(t)\right)  \notag \\
&&-\frac{1}{2}W^{\prime }(p)-\frac{1}{2}G^{\prime \prime }(q)p  \notag \\
&=&-\Phi _{\mathscr{C}}^{\prime }(q)-W^{\prime }(p)-G^{\prime \prime
}(q)p+e(t).  \label{eq:xxx}
\end{eqnarray}%
The system of equations has the desired form, and it remains only to choose $%
W$ and $G$ as in Eq. (\ref{eq:WG}) to specify the required resistance and
memristance.
\end{proof}


To see explicitly what the stochastic differential equations (SDEs) for $q_{t}$ and 
$p_{t}$ look like for this model we first observe that 
\begin{eqnarray*}
dq_{t}
&=&v^{q}(q_{t},p_{t})\,dt+\{q,F_{1}\}_{t}\,dB_{t}^{1}+\{q,F_{2}\}_{t}%
\,dB_{t}^{2} \\
&=&v^{q}(q_{t},p_{t})\,dt+cW^{\prime }(q_{t})\,dB_{t}^{1}+\frac{p_{t}}{\ell }%
\,dB_{t}^{2},
\end{eqnarray*}%
with a similar expression for the momentum. According to Theorem \ref{Thm:main}, the phase coordinate SDEs are
therefore 
\begin{eqnarray}
dq_{t} &=&\frac{p_{t}}{L_{0}}\,dt+cW^{\prime }(p_{t})\,dB_{t}^{1}+\frac{p_{t}%
}{\ell }\,dB_{t}^{2}, \\
dp_{t} &=&\left( -\Phi _{\mathscr{C}}^{\prime }\left( q_{t}\right) -\Psi _{%
\mathscr{R}}^{\prime }(p_{t})-\frac{1}{L_{0}}M\left( q_{t}\right)
p_{t}+e\left( t\right) \right) \ dt  \notag \\
&&-\frac{q_{t}}{c}\,dB_{t}^{1}-\ell G^{\prime }(q_{t})\,dB_{t}^{2}.
\end{eqnarray}%
We note that the dissipation here is 
\begin{equation}
\gamma (q,p)=W^{\prime \prime }(p)+G^{\prime \prime }(q) = \frac{\mathbb{R} (p) + M(q)}{L_0},
\end{equation}%
which is consistent to Eq. (\ref{eq:_secII_gamma}). 
Moreover, we also have non-trivial fluctuations, and for instance, 
\begin{equation}
dq_{t}\,dp_{t}=-\left( q_{t}W^{\prime }(p_{t})+p_{t}G^{\prime
}(q_{t})\right) \,dt.
\end{equation}%
The fluctuations and dissipation balance out to preserve the Poisson
structure.

It is a feature of canonical diffusions that the stochastic process $p_t$ is
no longer simply $\frac{1}{L_0} \dot{q}_t$. We would need a noise that was
more singular than Weiner processes if we wanted to have $dq_t = \frac{1}{L_0%
} p_t \, dt$ and still preserve the Poisson structure. In fact, to ensure
that the $dB^\alpha_t$ terms did not arise in the $dq_t$ SDE, we would need
that $\{ q,F_\alpha \}=0$ for each $\alpha$. But this would imply that each $%
F_\alpha$ is a function of $q$ only, in which case the drift terms $v^q$ and 
$v^p$ would be purely Hamiltonian, since we also have $\{ \{ p , F_\alpha \}
, F_\alpha \} \equiv 0$.

\subsection{Symplectic Noise Models}

In \cite{Gough_JSP_2015} the concept of canonically conjugate Wiener
processes was introduced. Here one extends the Poisson bracket structure to
include the noise so that in addition to each Wiener process $B_{k}$, which
we now relabel as $Q_{k}$, we have a conjugate process $P_{k}$ so the
collection $\left( Q_{k},P_{k}\right) $ have the statistics of 
independent Wiener processes, so that 
\begin{gather*}
dQ_{j}(t)\,dQ_{k}(t)=\delta _{jk}dt=dP_{j}(t)dP_{k}(t), \\
dQ_{j}(t)\,dP_{k}(t)=0=dP_{j}(t)\,dQ_{k}(t),
\end{gather*}%
but additionally satisfy 
\begin{equation}
\{Q_{j}(t),P_{k}(s)\}=\Gamma \,\delta _{jk}\,\mathrm{min}(t,s),
\label{eq:gamma}
\end{equation}%
where $\Gamma >0$. We shall restrict to just a single canonical pair $\left(
Q,P\right) $ in the following, although the generalization is
straightforward. For stochastic processes $X_{t}$ and $Y_{t}$ adapted to the
filtration generated by the Weiner processes $(Q_{s},P_{s})_{s\leq t}$, we have the
infinitesimal relations 
\begin{eqnarray*}
\{X_{t}\,dQ_{j}(t),Y_{t}\} &=&\{X_{t},Y_{t}\}\,dQ_{j}(t), \\
\{X_{t}\,dP_{j}(t),Y_{t}\} &=&\{X_{t},Y_{t}\}\,dP_{j}(t),
\end{eqnarray*}%
and 
\begin{eqnarray*}
&&\{X_{t}dQ_{j}(t),Y_{t}dP_{k}(t)\} \\
&=&\{X_{t},Y_{t}\}\,dQ_{j}(t)dP_{k}(t)+X_{t}Y_{t}\{dQ_{j}(t),dP_{k}(t)\} \\
&=&\delta _{jk}\,\Gamma \,X_{t}Y_{t}\,dt.
\end{eqnarray*}

Let us now consider a diffusion on phase space driven by a canonical pair of
Wiener processes: 
\begin{eqnarray}
dq_{t}  = w^{q}\left( q_{t},p_{t}\right) dt+\sigma ^{q}\left(
q_{t},p_{t}\right) dQ_{t}+\varsigma ^{q}\left( q_{t},p_{t}\right) dP_{t}, 
  \label{sys_1a} \\
dp_{t} = w^{p}\left( q_{t},p_{t}\right) dt+\sigma ^{p}\left(
q_{t},p_{t}\right) dQ_{t}+\varsigma ^{p}\left( q_{t},p_{t}\right) dP_{t}, 
\label{sys_1b}
\end{eqnarray}%
with the coefficients assumed to be Lipschitz, etc., so as to guarantee
existence and uniqueness of solution. The main result, Theorem 2, of \cite%
{Gough_JSP_2015} is stated below.

\begin{lemma}
\label{lem:Gough15} The diffusion (\ref{sys_1a})-(\ref{sys_1b}) on phase
space driven by a canonically conjugate pair of Wiener processes is
canonical for the full Poisson brackets if the vector fields $w,\sigma $ and 
$\varsigma $ are all Hamiltonian, say $w\left( \cdot \right) =\left\{ \cdot
,H\right\} $, $\sigma \left( \cdot \right) =\left\{ \cdot ,F\right\} $ and $%
\varsigma \left( \cdot \right) =\left\{ \cdot ,G\right\} $, in which case we
have 
\begin{equation*}
df_{t}=(\mathcal{L}f)_{t}\,dt+ \{f,F \}_{t}\,dQ (t)+\{f,G \}_{t}\,dP (t)
\end{equation*}
with generator 
\begin{equation*}
\mathcal{L}=\{\cdot ,H\}+\frac{1}{2} \big\{\{\cdot ,F\},F\big\}+ \frac{1}{2} \big\{%
\{\cdot ,G\},G\big\}  +u.\nabla ,
\end{equation*}
where $u$ is a vector field with 
\begin{equation} \label{eq:nabla_u}
\nabla .u=-\Gamma \{F,G\}.
\end{equation}
 In this case,
the dissipation function on phase space is given by 
\begin{equation*}
\gamma \left( q,p\right) =\left\{ \frac{\partial F}{\partial q},\frac{%
\partial F}{\partial p}\right\} +\left\{ \frac{\partial G}{\partial q},\frac{%
\partial G}{\partial p}\right\} +\Gamma \left\{ F,G\right\} .
\end{equation*}
\end{lemma}

\begin{example}
As an example, we consider the linear $LC$ model with the Hamiltonian $H_{0}$
defined in Eq. (\ref{eq:H0}) and 
\begin{equation*}
F=p,~G=-q.
\end{equation*}%
By Eq. (\ref{eq:nabla_u}) we have $\nabla .u=-\Gamma $. A particular solution is given by $u^q=0, u^p = -\Gamma p$. As  a result,  the equations of
motion are 
\begin{eqnarray}
dq_{t} &=&\frac{p_{t}}{L_{0}}dt+dQ_{t}, \\
dp_{t} &=&-\left( \frac{1}{C_{0}}q_{t}+\Gamma p_{t}-e(t)\right) dt+dP_{t}.
\end{eqnarray}%
Finally, it is easy to show that 
\begin{equation*}
\gamma \left( q,p\right) \equiv \Gamma .
\end{equation*}%
It is easy to see that $dq_{t}dp_{t}=0$. However, noises $(Q_t,P_t)$ do leave their imprint by adding the term $-\Gamma p_t$ to the momentum.  
\end{example}

In the above example, the proposed particular solution to $\nabla .u=-\Gamma $ can be given by
\begin{equation*}
u.\nabla =-\Gamma p\frac{\partial }{\partial p}.
\end{equation*}
The next two results show how to find the vector field $u$ in general.

\begin{lemma}
Given the family of twice-differentiable functions $( F_{\alpha }$, $%
G_{\alpha } )$, a particular solution to the equation
\begin{equation}\label{eq:nabla_u_alpha}
\nabla .u=-\Gamma \sum_\alpha \{F_\alpha,G_\alpha\}
\end{equation}
 is given by the vector
field $u_{0}$ with components 
\begin{eqnarray*}
u_{0}^{q} &=&\frac{1}{2}\Gamma \sum_{\alpha }\left( G_{\alpha }\frac{%
\partial F_{\alpha }}{\partial p}-F_{\alpha }\frac{\partial G_{\alpha }}{%
\partial p}\right) , \\
u_{0}^{p} &=&\frac{1}{2}\Gamma \sum_{\alpha }\left( F_{\alpha }\frac{%
\partial G_{\alpha }}{\partial q}-G_{\alpha }\frac{\partial F_{\alpha }}{%
\partial q}\right) .
\end{eqnarray*}
The general solution is then the particular solution plus an arbitrary
Hamiltonian vector field.
\end{lemma}

\begin{proof}
We verify directly that 
\begin{eqnarray*}
\frac{\partial u_{0}^{q}}{\partial q} &=&-\frac{1}{2}\Gamma \sum_{\alpha
}\left\{ F_{\alpha },G_{a}\right\} \\
&&+\frac{1}{2}\Gamma \sum_{\alpha }\left( G_{\alpha }\frac{\partial
^{2}F_{\alpha }}{\partial q\partial p}-F_{\alpha }\frac{\partial
^{2}G_{\alpha }}{\partial q\partial p}\right) , \\
\frac{\partial u_{0}^{p}}{\partial p} &=&-\frac{1}{2}\Gamma \sum_{\alpha
}\left\{ F_{\alpha },G_{a}\right\} \\
&&-\frac{1}{2}\Gamma \sum_{\alpha }\left( G_{\alpha }\frac{\partial
^{2}F_{\alpha }}{\partial p\partial q}-F_{\alpha }\frac{\partial
^{2}G_{\alpha }}{\partial p\partial q}\right),
\end{eqnarray*}
so adding the terms gives $\frac{\partial u_{0}^{q}}{\partial q}+\frac{%
\partial u_{0}^{p}}{\partial p}=-\Gamma \sum_{\alpha }\left\{ F_{\alpha
},G_{a}\right\} $ as required.
\end{proof}

\begin{corollary}
\label{cor:only} Given the family of twice-differentiable functions $%
F_{\alpha },G_{\alpha }$, a particular solution to the equation (\ref{eq:nabla_u_alpha}) is given by the
vector field $u$ with components 
\begin{eqnarray}
u^{q} &=&-\Gamma \sum_{\alpha }F_{\alpha }\frac{\partial G_{\alpha }}{%
\partial p}, \label{eq:uq_cor} \\
u^{p} &=&\Gamma \sum_{\alpha }F_{\alpha }\frac{\partial G_{\alpha }}{%
\partial q}. \label{eq:up_cor}
\end{eqnarray}
\end{corollary}

\begin{proof}
This can be verified by direct substitution into the equation (\ref{eq:nabla_u_alpha}). Alternatively, note that 
\begin{equation*}
u^{q}=u_{0}^{q}+\frac{\partial K}{\partial p},\quad u^{p}=u_{0}^{p}-\frac{%
\partial K}{\partial q},
\end{equation*}
with $K=\sum_{\alpha }F_{\alpha }G_{a}$.
\end{proof}

The following result gives the symplectic stochastic model which realizes the velocity fields (\ref{eq:drift_velocity}).

\begin{theorem}
\label{Thm:main2} A stochastic evolution with It\={o} drift velocity $%
(v^{q},v^{p})$ given by the phase velocity field (\ref{eq:drift_velocity})
can be achieved by the stochastic canonical model driven by a pair of
independent symplectic Wiener processes $(Q_{1},P_{1})$ and $(Q_{2},P_{2})$,
for the resistance and memristance components respectively, with the
following choices: 
\begin{eqnarray}
H &=&\frac{p^{2}}{2L_{0}}+\Phi _{\mathscr{C}}(q)-e(t)q,  \label{eq:hello2} \\
F_{1} &=&\varrho (p),\quad G_{1}=-q,  \label{eq:F1G1} \\
F_{2} &=&p,\quad G_{2}=-\mu (q),  \label{eq:F2G2}
\end{eqnarray}%
where $\varrho (p)=\frac{1}{\Gamma }\Psi _{\mathscr{R}}^{\prime }(p)$ and $%
\mu ^{\prime }(q)=\frac{1}{\Gamma L_{0}}M(q)$. Here we take $u$ to be the
vector field obtained from Corollary \ref{cor:only}. The symplectic
stochastic model of the system (\ref{eq:drift_velocity}) is 
\begin{eqnarray}
dq_{t} &=&v^{q}dt+\frac{\mathbb{R}(p)}{\Gamma L_{0}}dQ_{1}(t)+dQ_{2}(t),
\label{eq:sym_a} \\
dp_{t} &=&v^{p}dt-dP_{1}(t)-\frac{M^{\prime }(q)}{\Gamma L_{0}}dP_{2}(t),
\label{eq:sym_b}
\end{eqnarray}%
with It\={o} drift 
\begin{eqnarray}
v^{q} &=&\frac{1}{L_{0}}p,  \label{eq:vq} \\
v^{p} &=&-\Phi _{\mathscr{C}}^{\prime }\left( q\right) +e\left( t\right)
-\Psi _{\mathscr{R}}^{\prime }(p)-\frac{1}{L_{0}}M\left( q\right) p.
\label{eq:vp}
\end{eqnarray}
The negative divergence of the velocity field is 
\begin{equation}\label{eq:sym_gamma}
\gamma \left( q,p\right) = \frac{1}{L_{0}}\left( \mathbb{R}\left( p\right) +M\left( q\right) \right).
\end{equation}
\end{theorem}

\begin{proof}
 For the
choices $\left( F_{\alpha },G_{\alpha }\right) $ in the statement of the
Theorem, we have that the vector field obtained in Corollary \ref{cor:only}
will be 
\begin{eqnarray*}
u^{q} &=&0, \\
u^{p} &=&-\Gamma \left( \varrho \left( p\right) +p\mu ^{\prime }\left(
q\right) \right)  \\
&\equiv &-\Psi _{\mathscr{R}}^{\prime }(p)-\frac{p}{L_{0}}M\left( q\right) .
\end{eqnarray*}%
The It\={o} drift may be calculated from the generator in Lemma \ref{lem:Gough15} , for instance by $%
v^{q}=\mathcal{L}q$ we have
\begin{eqnarray*}
v^{q} &=&\frac{\partial H}{\partial p}+u^{q} \\
&+&\frac{1}{2}\sum_{\alpha }\left\{ \frac{\partial F_{\alpha }}{\partial p}%
,F_{\alpha }\right\} +\frac{1}{2}\sum_{\alpha }\left\{ \frac{\partial
G_{\alpha }}{\partial p},G_{\alpha }\right\} 
\end{eqnarray*}%
however $\sum_{\alpha }\left\{ \frac{\partial F_{\alpha }}{\partial p}%
,F_{\alpha }\right\} =\left\{ \varrho ^{\prime }\left( p\right) ,\varrho
\left( p\right) \right\} +\left\{ p,1\right\} =0$ and $\sum_{\alpha }\left\{ 
\frac{\partial G_{\alpha }}{\partial p},G_{\alpha }\right\} \equiv 0$, since
the $G_{\alpha }$ do not depend on $q$. So we have
\begin{eqnarray}
v^{q} =\frac{\partial H}{\partial p}+u^{q}=\frac{1}{L_{0}}p. \label{eq:sym_vq}
\end{eqnarray}%
 Similarly 
\begin{eqnarray*}
v^{p} &=&-\frac{\partial H}{\partial q}+u^{p} \\
&-&\frac{1}{2}\sum_{\alpha }\left\{ \frac{\partial F_{\alpha }}{\partial q}%
,F_{\alpha }\right\} -\frac{1}{2}\sum_{\alpha }\left\{ \frac{\partial
G_{\alpha }}{\partial q},G_{\alpha }\right\} .
\end{eqnarray*}%
but again $\sum_{\alpha }\left\{ \frac{\partial F_{\alpha }}{\partial q}%
,F_{\alpha }\right\} \equiv 0$, as the $F_{\alpha }$ do not depend on $p$,
and $\sum_{\alpha }\left\{ \frac{\partial G_{\alpha }}{\partial q},G_{\alpha
}\right\} =\left\{ 1,q\right\} +\left\{ \mu ^{\prime }\left( q\right) ,\mu
\left( q\right) \right\} =0$. Therefore, we have 
\begin{eqnarray}
v^{p} &=&-\frac{\partial H}{\partial q}+u^{p} \nonumber \\
&=&-\Phi _{\mathscr{C}}^{\prime }\left( q\right) +e\left( t\right) -\Psi _{%
\mathscr{R}}^{\prime }(p)-\frac{1}{L_{0}}M\left( q\right) p. \label{eq:sym_vp}
\end{eqnarray}%
Equations (\ref{eq:sym_vq})-(\ref{eq:sym_vp}) are the desired form of the It\={o} drift (\ref{eq:vq})-(\ref{eq:vp}).
The stochastic evolution (\ref{eq:sym_a})-(\ref{eq:sym_b}) follows
immediately. Finally, We note that the dissipation will be 
\begin{eqnarray}
\gamma \left( q,p\right)  &=&\Gamma \sum_{\alpha }\left\{ F_{\alpha
},G_{\alpha }\right\}   \notag \\
&=&\Gamma \left( \varrho ^{\prime }(p)+\mu ^{\prime }\left(
q\right) \right)   \notag \\
&=&\frac{1}{L_{0}}\left( \mathbb{R}\left( p\right) +M\left( q\right) \right) 
\label{eq:gamma_sym}
\end{eqnarray}%
which is the correct form, (recall that $\Psi _{\mathscr{R}%
}^{\prime \prime }(p)\equiv \mathbb{R}\left( p\right) /L_{0}$, cf. Eq. (\ref{eq:_secII_gamma})). Eq. (\ref{eq:sym_gamma}) is established.
\end{proof}


\section{Quantum Stochastic Models} \label{sec:Quantum} 

To quantize, we replace $q$ and $p$ by operators
satisfying the commutation relations 
\begin{equation*}
\left[ \hat{q},\hat{p}\right] =i\hbar .
\end{equation*}%
In the case of an $LC$ circuit driven by a classical emf $e$ we have the
Hamiltonian 
\begin{equation*}
\hat{H}_{0}=\frac{1}{2L_{0}}\hat{p}^{2}+\frac{1}{2C_{0}}\hat{q}^{2}-e(t)\hat{%
q}.
\end{equation*}%
We may introduce annihilation operators, defined as 
\begin{equation*}
\hat{a}=\left( 2\hbar \omega L_{0}\right) ^{-1/2}\left( \omega _{0}L_{0}\hat{%
q}+ip\right)
\end{equation*}%
and so 
\begin{eqnarray*}
\hat{q} &=&\sqrt{\frac{\hbar \omega _{0}C_{0}}{2}}\left( \hat{a}+\hat{a}%
^{\ast }\right) , \\
\hat{q} &=&i\sqrt{\frac{\hbar \omega _{0}L_{0}}{2}}\left( \hat{a}^{\ast }-%
\hat{a}\right) .
\end{eqnarray*}%
The Hamiltonian is then 
\begin{equation}
\hat{H}_{0}=\hbar \omega _{0}\left( \hat{a}^{\ast }\hat{a}+\frac{1}{2}%
\right) -\sqrt{\frac{\hbar \omega _{0}C}{2}}\left( \hat{a}+\hat{a}^{\ast
}\right) e\left( t\right) .  \label{eq:Ham_osc}
\end{equation}

While the formalism is identical to the quantum mechanical oscillator, we
may give an entirely different interpretation physically. The state $%
|n\rangle $ for the quantized circuit describes the situation where there
are $n$ quanta (photons) in the circuit. In this case the number operator $%
\hat{N}=\hat{a}^{\ast }\hat{a}$ is the observable corresponding to the
number of photons in the circuit. Note that we may have zero photons, and
this corresponds to a physical state $|0\rangle $. The source term in the
Hamiltonian is proportional to $e\left( t\right) $ and involves the creation 
$\hat{a}^{\ast }$ and annihilation $\hat{a}$ of photons.

We are in the situation that there exists a theory of quantum stochastic
integration generalizing the It\={o} calculus. The mathematical theory was
developed in 1984 by Hudson and Parthasarathy, \cite{HP84}, though was also
derived on a physical basis for modelling physical noise in quantum photonic
models by Gardiner and Collett in 1985, \cite{GC85}. In essence, we have a quantum system
with underlying Hilbert space $\mathfrak{h}$ which is in interaction with an
infinite environment modelled as a quantum field, \cite{HP84}, \cite{Pa92}, \cite{GarZol00}. Without
going too far into the details we have quantum white noise fields $\hat{a}%
_{\alpha }\left( t\right) $ which are operators satisfying commutation
relations 
\begin{equation*}
\left[ \hat{a}_{\alpha }\left( t\right) ,\hat{a}_{\beta }^{\ast }\left(
s\right) \right] =\delta _{\alpha \beta }\,\delta \left( t-s\right) .
\end{equation*}

We also fix the vacuum state which is the unique vector such that $\hat{a}%
_{\alpha }\left( t\right) |\Omega \rangle =0$.

The integrated fields $\hat{A}_{\alpha }\left( t\right) =\int_{0}^{t}\hat{a}%
_{\alpha }\left( s\right) ds$ and $\hat{A}_{\alpha }^{\ast }\left( t\right)
=\int_{0}^{t}\hat{a}_{\alpha }^{\ast }\left( s\right) ds$ are well defined
operators on Bose Fock space $\mathfrak{F}$ satisfying 
\begin{equation*}
\left[ \hat{A}_{\alpha }\left( t\right) ,\hat{A}_{\beta }^{\ast }\left(
s\right) \right] =\delta _{\alpha \beta }\min \left\{ t,s\right\} \text{.}
\end{equation*}%
The quadrature processes are defined by 
\begin{equation*}
\hat{Q}_{\alpha }\left( t\right) =\hat{A}_{\alpha }\left( t\right) +\hat{A}%
_{\alpha }^{\ast }\left( t\right) ,\hat{P}_{\alpha }\left( t\right) =\frac{1%
}{i}\left( \hat{A}_{\alpha }\left( t\right) -\hat{A}_{\alpha }^{\ast }\left(
t\right) \right) ,
\end{equation*}
both of which are self-commuting operator-valued processes, that is, 
\begin{equation*}
\left[ \hat{Q}_{\alpha }\left( t\right) ,\hat{Q}_{\beta }\left( s\right) %
\right] =0=\left[ \hat{P}_{\alpha }\left( t\right) ,\hat{P}_{\beta }\left(
s\right) \right] .
\end{equation*}%
In the vacuum state of the noise, they both have the statistics of a
standard Wiener process. However, the quadrature processes do not commute
and we have instead 
\begin{equation}
\left[ \hat{Q}_{\alpha }\left( t\right) ,\hat{P}_{\beta }\left( s\right) %
\right] =2i\delta _{\alpha \beta }\min \left\{ t,s\right\} .  \label{eq:CCC}
\end{equation}%
Hudson and Parthasarathy developed a theory of quantum stochastic
integration wrt. these processes. This involves the following nontrivial
product of It\={o} increments 
\begin{equation*}
d\hat{A}_{\alpha }\left( t\right) d\hat{A}_{\beta }^{\ast }\left( t\right)
=\delta _{\alpha \beta }dt.
\end{equation*}%
They show that a quantum stochastic process $\hat{U}\left( t\right) $ can be
defined on the joint system+noise space $\mathfrak{h}\otimes \mathfrak{F}$
by 
\begin{eqnarray*}
d\hat{U}\left( t\right) &=&\big(\sum_{\alpha }\hat{L}_{\alpha }d\hat{A}%
_{\alpha }^{\ast }\left( t\right) -\sum_{\alpha }\hat{L}_{\alpha }^{\ast }d%
\hat{A}_{\alpha }\left( t\right) \\
&&-\frac{1}{2}\sum_{\alpha }^{\ast }\hat{L}_{\alpha }^{\ast }L_{\alpha }dt-%
\frac{i}{\hbar }\hat{H}_{0}dt\big)\hat{U}\left( t\right) , \\
\hat{U}(0) &=&I_{\mathfrak{h}}\otimes I_{\mathfrak{F}},
\end{eqnarray*}%
and that the process is unitary. (Technically they require the system
operators $\hat{H}_{0}=\hat{H}_{0}^{\ast }$ and $\hat{L}_{\alpha }$ to be
bounded, however, the theory extends to unbounded coefficients). The
dynamical evolution of a system observable $\hat{X}$ is given by 
\begin{equation*}
j_{t}\left( \hat{X}\right) =\hat{U}\left( t\right) ^{\ast }\left( \hat{X}%
\otimes I_{\mathfrak{F}}\right) \hat{U}\left( t\right),
\end{equation*}%
and one can deduce the following dynamical equations of motion: 
\begin{eqnarray}
dj_{t}\left( \hat{X}\right) &=&\sum_{\alpha }j_{t}\left( \left[ \hat{X},\hat{%
L}_{\alpha }\right] \right) d\hat{A}_{\alpha }^{\ast }\left( t\right)  \notag
\\
&&+\sum_{\alpha }j_{t}\left( \left[ \hat{L}_{\alpha }^{\ast },\hat{X}\right]
\right) d\hat{A}_{\alpha }\left( t\right)  \notag \\
&&+j_{t}\left( \mathcal{L}\hat{X}\right) dt  \label{QSDE:X}
\end{eqnarray}%
where the generator takes the form 
\begin{eqnarray*}
\mathcal{L}\hat{X} &=&\frac{1}{2}\sum_{\alpha }\left[ \hat{L}_{\alpha
}^{\ast },\hat{X}\right] \hat{L}_{\alpha }+\frac{1}{2}\sum_{\alpha }\hat{L}%
_{\alpha }^{\ast }\left[ \hat{X},\hat{L}_{\alpha }\right] \\
&&+\frac{1}{i\hbar }\left[ \hat{X},\hat{H}_{0}\right] .
\end{eqnarray*}%
This is the well-known GKS-Lindblad generator from the theory of quantum
dynamical semi-groups, \cite{Lindblad},  \cite{GKS}, \cite{Go05}, \cite{Go06}. The Hudson-Parthasarathy theory therefore offers a unitary dilation of such quantum Markov semi-groups.

We now show how to realize the analogue of the It\={o} drift corresponding
to the phase velocity field (\ref{eq:drift_velocity}). Note that in this
case the terms $pM(q)$ is ambiguous as $\hat{p}$ and $M(\hat{q})$ generally
do not commute. We naturally interpret this as the symmetrically ordered
term $\frac{1}{2}\hat{p}M(\hat{q})+\frac{1}{2}M(\hat{q})\hat{p}$. Motivated
by the previous Theorem \ref{Thm:main2} giving the construction of such a
diffusion in the classical setting with two independent canonical pairs of
symplectic Wiener noise, we now establish the corresponding quantum analogue.

\begin{theorem}
\label{Thm:main3} Given the function $\Phi _{\mathscr{C}}$, $\Psi _{%
\mathscr{R}}$ and $M$ we obtain the It\={o} drift terms 
\begin{eqnarray*}
v^{q} &=&\mathcal{L}\hat{q}=\frac{\hat{p}}{L_{0}}, \\
v^{p} &=&\mathcal{L}\hat{p}=-\Phi _{\mathscr{C}}^{\prime }\left( \hat{q}%
\right) -\Psi _{\mathscr{R}}^{\prime }(\hat{p}) \\
&&-\frac{1}{2L_{0}}\left( M\left( \hat{q}\right) \hat{p}+\hat{p}M\left( \hat{%
q}\right) \right) +e\left( t\right) ,
\end{eqnarray*}%
for the choice $\hat{H}=\hat{H}_{0}+\hat{K}$ where $\hat{H}_{0}$ is the
Hamiltonian (\ref{eq:Ham_osc}) and 
\begin{eqnarray*}
\hat{K} &=&\frac{1}{2}\left[ f(\hat{p})\hat{q}+\hat{q}f(\hat{p})\right] , \\
&&+\frac{1}{2}\left[ \hat{p}g\left( \hat{q}\right) +g\left( \hat{q}\right) 
\hat{p}\right]
\end{eqnarray*}%
and coupling terms 
\begin{eqnarray}
\hat{L}_{1} &=&\hat{q}+i\frac{1}{\hbar }f\left( \hat{p}\right) ,
\label{eq:L1} \\
\hat{L}_{2} &=&\frac{1}{\hbar }g\left( \hat{q}\right) +i\hat{p},
\label{eq:L2}
\end{eqnarray}%
with the functions $f$ and $g$ given by 
\begin{eqnarray*}
f\left( \hat{p}\right) &=&\frac{1}{2}\Psi _{\mathscr{R}}^{\prime }(\hat{p}),
\\
g^{\prime }\left( \hat{q}\right) &=&\frac{1}{2L_{0}}M\left( \hat{q}\right) .
\end{eqnarray*}%
The quantum stochastic model of the system (\ref{eq:drift_velocity}) is 
\begin{eqnarray}
dj_{t}(\hat{q}) = j_{t}(v^{q})dt-j_{t}(f^{\prime }(\hat{p}))d\hat{Q}%
_{1}(t)-\hbar d\hat{Q}_{2}(t), \\
dj_{t}(\hat{p}) = j_{t}(v^{p})dt-\hbar d\hat{P}_{1}(t)-j_{t}(g^{\prime }(%
\hat{q}))d\hat{P}_{2}(t).
\end{eqnarray}
\end{theorem}

\begin{proof}
We have for instance 
\begin{equation*}
\frac{1}{i\hbar }\left[ \hat{q},\hat{H}\right] =\frac{\hat{p}}{L_{0}}+\frac{1%
}{2}\left[ f^{\prime }\left( \hat{p}\right) \hat{q}+\hat{q}f^{\prime }\left( 
\hat{p}\right) \right] +g\left( \hat{q}\right) ,
\end{equation*}%
and 
\begin{eqnarray*}
\frac{1}{2}\left[ \hat{L}_{1}^{\ast },\hat{q}\right] \hat{L}_{1}+\frac{1}{2}%
\hat{L}_{1}^{\ast }\left[ \hat{q},\hat{L}_{1}\right] &=&-\frac{1}{2}\left[
f^{\prime }\left( \hat{p}\right) \hat{q}+\hat{q}f^{\prime }\left( \hat{p}%
\right) \right] \\
\frac{1}{2}\left[ \hat{L}_{2}^{\ast },\hat{q}\right] \hat{L}_{2}+\frac{1}{2}%
\hat{L}_{2}^{\ast }\left[ \hat{q},\hat{L}_{2}\right] &=&-g\left( \hat{q}%
\right)
\end{eqnarray*}%
which combine to give $\mathcal{L}\hat{q}=\frac{\hat{p}}{L_{0}}$. Similarly,
we have 
\begin{eqnarray*}
\frac{1}{i\hbar }\left[ \hat{p},\hat{H}\right] &=&-\Phi _{\mathscr{C}%
}^{\prime }\left( \hat{q}\right) +e(t) \\
&&-f\left( \hat{p}\right) -\frac{1}{2}\left[ g^{\prime }\left( \hat{q}%
\right) \hat{p}+\hat{p}f^{\prime }\left( \hat{q}\right) \right] ,
\end{eqnarray*}%
and 
\begin{eqnarray*}
\frac{1}{2}\left[ \hat{L}_{1}^{\ast },\hat{p}\right] \hat{L}_{1}+\frac{1}{2}%
\hat{L}_{1}^{\ast }\left[ \hat{p},\hat{L}_{1}\right] &=&-f(\hat{p}) \\
\frac{1}{2}\left[ \hat{L}_{2}^{\ast },\hat{q}\right] \hat{L}_{2}+\frac{1}{2}%
\hat{L}_{2}^{\ast }\left[ \hat{q},\hat{L}_{2}\right] &=&-\frac{1}{2}%
(g^{\prime }(\hat{q})\hat{p}+\hat{p}g^{\prime }(\hat{q}))
\end{eqnarray*}%
which combine to give 
\begin{eqnarray*}
\mathcal{L}\hat{p} &=&-\Phi _{\mathscr{C}}^{\prime }\left( \hat{q}\right)
+e(t) \\
&&-2f\left( \hat{p}\right) -\left[ g^{\prime }\left( \hat{q}\right) \hat{p}+%
\hat{p}f^{\prime }\left( \hat{q}\right) \right]
\end{eqnarray*}%
so that the choices $f\left( \hat{p}\right) =\frac{1}{2}\Psi _{\mathscr{R}%
}^{\prime }(\hat{p})$ and $g^{\prime }\left( \hat{q}\right) =\frac{1}{2L_{0}}%
M\left( \hat{q}\right) $ give the desired form.
\end{proof}

\begin{remark}
If we choose $\Gamma =2$ in Eq. (\ref{eq:gamma}), (notice that there exists
factor 2 for the canonical commutator in Eq. (\ref{eq:CCC})) then Eqs. (\ref%
{eq:F1G1})-(\ref{eq:F2G2}) become%
\begin{eqnarray*}
F_{1} &=&\frac{1}{2}\Psi _{R}^{^{\prime }}(p),\ \ \ F_{2}=p, \\
G_{1} &=&q,\ \ \ G_{2}=\frac{1}{2L_{0}}M^{^{\prime }}(q).
\end{eqnarray*}%
As a result,%
\begin{eqnarray}
G_{1}+\frac{i}{\hbar }F_{1} &=&q+\frac{i}{\hbar }\frac{1}{2}\Psi
_{R}^{^{\prime }}(p),  \label{eq:cq_1a} \\
\frac{1}{\hbar }G_{2}+iF_{2} &=&\frac{1}{\hbar }\frac{1}{2L_{0}}M^{^{\prime
}}(q)+ip.  \label{eq:cq_1b}
\end{eqnarray}%
Replacing $q$ and $p$ by $\hat{q}$ and $\hat{p}$ respectively in Eqs. (\ref%
{eq:cq_1a})-(\ref{eq:cq_1b}) we get exactly operators $\hat{L}_{1}$ and $%
\hat{L}_{2}$ in Eqs. (\ref{eq:L1})-(\ref{eq:L2}).
\end{remark}

\section{Approximations} \label{sec:appro}

In Sections \ref{sec:Stochastic_Models} and \ref{sec:Quantum} we have proposed stochastic models for dissipative electronic circuits.  A natural question to ask is whether the noisy dynamics can be approximated by systems which are lossless, or more specifically, Hamiltonian. In this section we give an affirmative answer to this question. 

We remark that we may find approximation schemes to Wiener noise. For
instance, it is possible to construct processes $B^{(n)} (t)$ that are
continuously differentiable in the time $t$ variable and which converge
almost surely to a Wiener process $B(t)$ uniformly in $t$ in compacts. The
random Hamiltonian 
\begin{equation*}
\Upsilon^{(n)} (q,p) = H(q,p) + F(q,p) \, B^{(n)} (t),
\end{equation*}
generates the following equations of motion for $x^{(n)} (t) = ( q^{(n)}(t)
, p^{(n)} (t) )$ 
\begin{eqnarray*}
\frac{d}{dt} q^{(n)}(t) &=& \frac{\partial H}{\partial p} |_{x^{(n)} (t) } + 
\frac{\partial F}{\partial p} |_{x^{(n)} (t) } \, B^{(n)} (t), \\
\frac{d}{dt} p^{(n)}(t) &=& -\frac{\partial H}{\partial q} |_{x^{(n)} (t) }
- \frac{\partial F}{\partial q} |_{x^{(n)} (t) } \, B^{(n)} (t),
\end{eqnarray*}
By the Wong-Zakai theorem \cite{Wong_Zakai} , the process $x^{(n)} (t)$ then
converges uniformly in $t$ on compact almost surely to the solution of the
Stratonovich SDEs 
\begin{eqnarray*}
d q (t) &=& \frac{\partial H}{\partial p} |_{x (t) } dt+ \frac{\partial F}{%
\partial p} |_{x (t) } \circ dB(t), \\
dp (t) &=& -\frac{\partial H}{\partial q} |_{x (t) } dt - \frac{\partial F}{%
\partial q} |_{x (t) } \circ dB (t),
\end{eqnarray*}
This is the single noise case of Theorem \ref{thm:JMP1999}. The multiple
noise case is the obvious generalization.

In order to describe the symplectic noise, we set about developing an
approximation using lossless circuits. In Figure \ref{fig:lossless_TL_gray}
below we approximate a transmission line as an assembly of lossless $LC$
circuits.

\begin{figure}[hbp]
\centering
\includegraphics[width=0.40\textwidth]{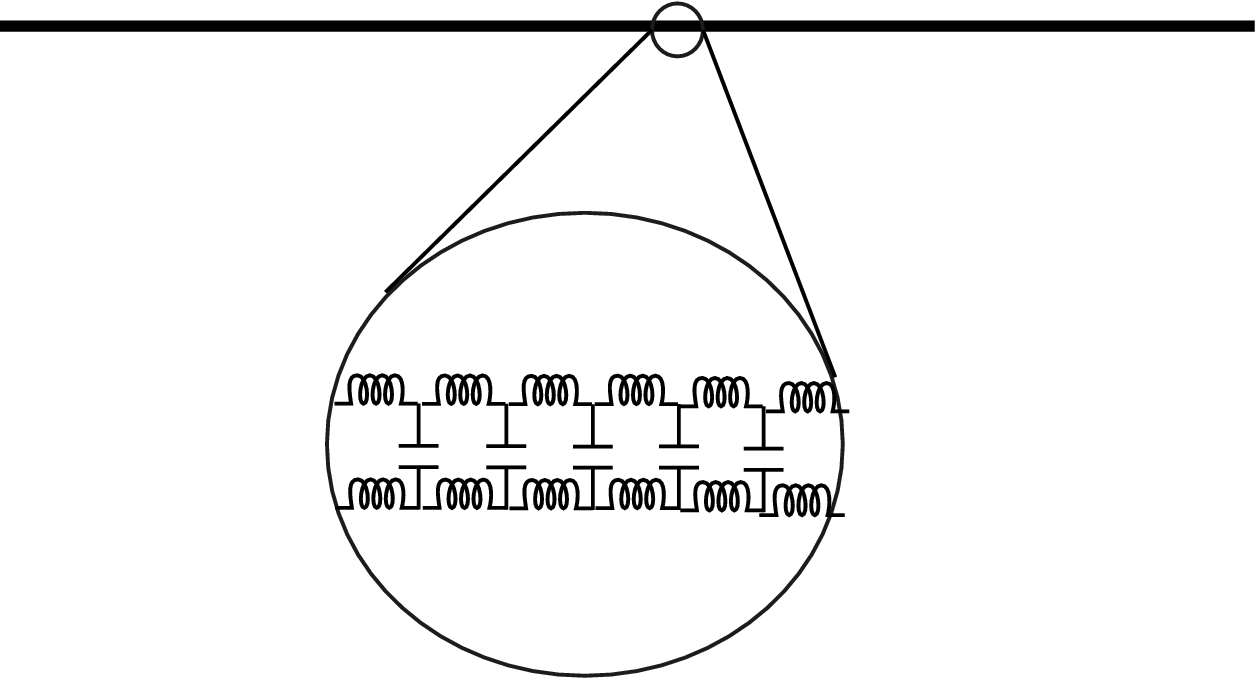}
\caption{A continuous transmission line approximated by shunted LC circuits.}
\label{fig:lossless_TL_gray.png}
\end{figure}

Let $q_k ,p_k$ be the canonical variables satisfying the relations $\{ q_j
,q_k \} =0 =\{p_j , p_k \}$ and $\{q_j ,p_k \} = \delta_{jk}$. For $t>0$ and
integer $N>0$, let us set 
\begin{eqnarray*}
Q^{(N)}(t) = \frac{1}{\sqrt{N}} \sum_k^{ \left\lfloor N t\right\rfloor} q_k,
\quad P^{(N)}(t) = \frac{1}{\sqrt{N}} \sum_k^{ \left\lfloor N t
\right\rfloor}p_k ,
\end{eqnarray*}
where $\left\lfloor x \right\rfloor$ means rounding down to the nearest
integer value. We evidently have $\{ Q^{(N)}(t) , P^{(N)}(t) \} = \frac{1}{N}
\left\lfloor N t \right\rfloor$ which evidently converges to $t$ as $N \to
\infty$. Let us take the energy of the $k$th circuit to be $\frac{1}{2L_0}
p^2_k + \frac{1}{2C_0} q^2_k$ and consider the canonical ensemble
corresponding to circuits in thermal equilibrium at temperature $T$. (In
practice, for $t$ and $N$ finite we only need a finite number in the
assembly.) The pairs of variables $(q_k ,p_k)$ are then independent and
identically distributed with mean zero and variance $Var(q_k ) = C_0 k_B T,
\, Var(p_k ) = L_0 k_B T$ and covariance $Cov(q_k ,p_k) = 0$, where $k_B$ is the Boltzmann constant. By the central
limit effect we see that the pair $(Q^{(N)} , P^{(N)} )$ converge to
independent Wiener processes with temperature dependent variances (which we
can always absorb). The result is a limit symplectic noise obtained as a
limit of thermalized lossless oscillator circuits.

There is a related result for quantum stochastic evolutions. We start with
the Schr\"{o}dinger equation 
\begin{equation*}
i\hbar \frac{d}{dt}\hat{U}^{\left( N\right) }=\hat{\Upsilon}^{\left(
N\right) }\left( t\right) \,\hat{U}^{\left( N\right) }\left( t\right)
\end{equation*}%
with the time dependent Hamiltonian 
\begin{equation*}
\hat{\Upsilon}^{\left( N\right) }\left( t\right) =\hat{E}\otimes \hat{a}%
^{\left( N\right) }\left( t\right) ^{\ast }+\hat{E}^{\ast }\otimes \hat{a}%
^{\left( N\right) }\left( t\right) +\hat{H},
\end{equation*}%
where we have regular reservoir field operators satisfying commutation
relations 
\begin{equation*}
\left[ \hat{a}^{\left( N\right) }\left( t\right) ,\hat{a}^{\left( N\right)
}\left( s\right) ^{\ast }\right] =g^{\left( N\right) }\left( t-s\right) .
\end{equation*}%
We fix the vacuum state $|\Omega \rangle $ for the reservoir so that $\hat{a}%
^{\left( N\right) }\left( t\right) |\Omega \rangle =0$. In the limit $%
N\rightarrow \infty $ we assume that $g^{\left( N\right) }\left( \tau
\right) \rightarrow \delta \left( \tau \right) $ in distribution with $%
\int_{0}^{\infty }g^{\left( N\right) }\left( \tau \right) d\tau =\frac{1}{2}$%
. Then we find the limit 
\begin{eqnarray*}
&&\lim_{N\rightarrow \infty }\langle \phi _{1}\otimes e^{\int f_{1}\left(
u\right) \hat{a}^{\left( N\right) }\left( u\right) ^{\ast }du}\Omega \\
&&|\hat{U}^{\left( N\right) }\left( t\right) \,\phi _{2}\otimes e^{\int
f_{2}\left( v\right) \hat{a}^{\left( N\right) }\left( v\right) ^{\ast
}dv}\Omega \rangle \\
&=&\langle \phi _{1}\otimes e^{\int f_{1}\left( u\right) dA^{\ast }\left(
u\right) }\Omega |\hat{U}\left( t\right) \,\phi _{2}\otimes e^{\int
f_{2}\left( v\right) aA^{\ast }(v)}\Omega \rangle
\end{eqnarray*}%
for arbitrary $\phi _{1},\phi _{2}$ in the system Hilbert space $\mathfrak{h}
$ and $L^{2}$ functions $f_{1},f_{2}$ where $\hat{U}$ is the solution to the
quantum stochastic differential equation (QSDE) 
\begin{eqnarray*}
d\hat{U}\left( t\right) &=&\bigg\{\hat{L}d\hat{A}^{\ast }\left( t\right) -%
\hat{L}^{\ast }d\hat{A}\left( t\right) \\
&&-\left( \frac{1}{2}\hat{L}^{\ast }\hat{L}+\frac{i}{\hbar }\hat{H}\right) dt%
\bigg\}\hat{U}\left( t\right) ,
\end{eqnarray*}%
with $\hat{L}=\hat{E}/i\hbar $. (In both cases we start with the initial
condition that the unitary is the identity.) This limit is clearly the
single noise channel quantum stochastic evolution considered in Section \ref%
{sec:Quantum}. A similar result holds for the Heisenberg equations, as well
as generalizations to thermal states of the reservoir.

\section{Conclusion} \label{sec:conclusion}

In this paper we have proposed stochastic models for electric circuits that may contain memristors, both the classical and quantum versions of noisy dynamics have been obtained. 
Preservation of the canonical structure was used as a guiding principle and the resulting theory allows for approximations schemes using Hamiltonian systems.
Future research includes the application of the proposed stochastic models to more general memristive electric circuits and making deeper connections with the
underlying statistical mechanical derivations.

\section*{Acknowledgment}
This work was supported by the Royal Academy of Engineering's UK-China Research Exchange Scheme, and JG gratefully acknowledges the
kind hospitality of Hong Kong Polytechnic during visit in the Fall of 2014.

\begin{IEEEbiography}{John E. Gough}
was born in Drogheda, Ireland, in 1967. He received the B.Sc. and M.Sc. in Mathematical Sciences
and the Ph.D. degree in Mathematical Physics from the National University of Ireland, Dublin, in 1987, 1988 and 1992
respectively. He was reader in Mathematical Physics at the Department of Mathematics and Computing, 
Nottingham-Trent University, up until 2007 when he joined the Institute of Mathematics and Physics at Aberystwyth 
University where he is now chair of Mathematical Physics. He has held visiting positions at the University of Rome Tor Vergata,
the Australian National University, EPF Lausanne, UC Santa Barbara and the Hong Kong Polytechnic University. His research interests include quantum probability, 
measurement and control of open quantum dynamical systems, and quantum feedback networks.
\end{IEEEbiography}

\begin{IEEEbiography}{Guofeng Zhang}
received his B.Sc. and M.Sc. degree from Northeastern University, Shenyang, China, in
1998 and 2000 respectively. He received a Ph.D. degree in
Applied Mathematics from the University of Alberta, Edmonton, Canada, in 2005. During 2005–2006, he was a
Postdoc Fellow in the Department of Electrical and Computer Engineering at the University of Windsor, Windsor,
Canada. He joined the School of Electronic Engineering
of the University of Electronic Science and Technology of
China, Chengdu, Sichuan, China, in 2007. From April 2010
to December 2011 he was a Research Fellow in the School
of Engineering of the Australian National University. He is currently an Assistant
Professor in the Department of Applied Mathematics at the Hong Kong polytechnic
University. His research interests include quantum control, sampled-data control
and nonlinear dynamics.
\end{IEEEbiography}

\end{document}